\documentclass[11pt]{article}

\usepackage[english]{babel}
\usepackage[T1]{fontenc}
\usepackage[utf8x]{inputenc}

\usepackage{amsthm}
\usepackage{mathrsfs}
\usepackage{amsfonts}
\usepackage{amssymb}
\usepackage{amsmath}
\usepackage{todonotes}
\usepackage{authblk}

\usepackage{color}
\usepackage{graphicx}
\usepackage{graphics}
\usepackage{xspace}

\usepackage{fullpage}

\usepackage[ruled,vlined]{algorithm2e}

\def\cC{\mathcal{C}}

\def\cF{\mathcal{F}}

\def\cN{\mathcal{N}}

\def\cP{\mathcal{P}}

\def\cR{\mathcal{R}}

\def\cU{\mathcal{U}}

\newcommand{\del}{\textsc{Del}}

\newcommand{\candidates}{\textsc{Candidates}}
\newcommand{\cc}{c}
\newcommand{\Fmin}{\mathcal{F}_\text{min}}
\newcommand{\can}{\textsc{Can}}
\newcommand{\flip}{\textsc{Flip}}
\newcommand{\successor}{\textsc{Succ}}
\newcommand{\prox}[2]{#1 \tilde\cap #2}
\newcommand{\parent}{\textsc{Parent}}
\newcommand{\children}{\textsc{Children}}
\newcommand{\neig}{\textsc{Neighbours}}
\newcommand{\Fmax}{\mathcal{F}_\text{max}}
\newcommand{\Next}{\textsc{Next}}

\newtheorem{thm}{Theorem}
\newtheorem{theorem}[thm]{Theorem}

\newtheorem{lemma}[thm]{Lemma}

\newtheorem{definition}[thm]{Definition}

\newtheorem{corollary}[thm]{Corollary}

\newtheorem{proposition}[thm]{Proposition}

\newtheorem{rmq}[thm]{Remark}

\theoremstyle{remark}
\newtheorem*{claim}{Claim}
\newenvironment{proofclaim}[1][]{\noindent\emph{Proof.\ } {}{#1}{}}{\hfill $\Diamond$\vspace{1em}}

\author[1]{Caroline Brosse}

\author[1]{Vincent Limouzy}

\author[2]{Arnaud Mary}

\affil[1]{Universit\'e Clermont Auvergne, Clermont
              Auvergne INP, CNRS, Mines Saint-Etienne, LIMOS, F-63000
              CLERMONT-FERRAND, France}
\affil[2]{Universit\'e de Lyon, F-69000, Lyon, CNRS, UMR5558, LBBE, F-69622 Villeurbanne, France / INRIA Grenoble Rh\^one-Alpes - ERABLE}

\title{Polynomial delay algorithm for minimal chordal completions}

\begin{document}

\maketitle
\thispagestyle{empty}
\DontPrintSemicolon
\SetProcNameSty{sc} 
\SetFuncSty{sc}


\maketitle
\thispagestyle{empty}

\begin{abstract}
Motivated by the problem of enumerating all tree decompositions of a graph, we consider in this article the problem of listing all the minimal chordal completions of a graph. In \cite{carmeli2020} (\textsc{Pods 2017}) Carmeli \emph{et al.} proved that all minimal chordal completions or equivalently all proper tree decompositions of a graph can be listed in incremental polynomial time using exponential space.  The total running time of their algorithm is quadratic in the number of solutions and the existence of an algorithm whose complexity depends only linearly on the number of solutions remained open. We close this question by providing a polynomial delay algorithm to solve this problem which, moreover, uses polynomial space.

Our algorithm relies on \emph{Proximity Search}, a framework recently introduced by Conte \emph{et al.} \cite{conte-uno2019} (\textsc{Stoc 2019}) which has been shown powerful to obtain polynomial delay algorithms, but generally requires exponential space. In order to obtain a polynomial space algorithm for our problem, we introduce a new general method called \emph{canonical path reconstruction} to design polynomial delay and polynomial space algorithms based on proximity search.
\end{abstract}


\section{Introduction} 
\label{sec:introduction}



Since its introduction by Dirac \cite{Dirac61}, the class of chordal 
graphs received a lot of attention.
The many interesting properties of chordal graphs (\emph{e.g.} a 
linear number of maximal cliques and 
minimal separators, a useful intersection model, a specific elimination
ordering to cite a few) lead to the design of efficient algorithms 
for problems that are usually difficult on general graphs.
On top of that, chordal graphs are closely connected to an important
graph parameter called \emph{treewidth} and its associated \emph{tree-decomposition} introduced independently by Halin \cite{Halin76}
and Robertson \& Seymour \cite{RobertsonS84}.
 
Treewidth has played an important role in algorithmics for 
the last forty years, since its introduction and popularization by Robertson \& Seymour.
This popularity is deserved: given a tree decomposition of small width,
many problems that are usually hard become solvable in polynomial time.

From the structure of maximal cliques known for chordal graphs, it is well-known
that an optimal tree-decomposition can be computed in linear time on this class.
For general graphs, one way to define the treewidth is to find the smallest value $k$ 
such that the graph is a subgraph of a $k$-tree, that is to say a subgraph of 
a chordal graph with maximum clique size at most $k+1$.
Unfortunately, it was shown by Arnborg, Corneil \& Proskurowski \cite{ArnborgCP87}
that determining whether a graph is a partial $k$-tree
(\emph{i.e.} a subgraph of a $k$-tree) is NP-Complete.
In a different direction, the \emph{minimum fill-in problem} asks to find the minimum 
number of edges to add to the graph in order to turn it into a chordal graph.
Yannakakis \cite{Yannakakis81} proved that the minimum fill-in problem is NP-hard.

%

Since the aforementioned problems are intractable, relaxations of these problems have 
been considered.
The problem of computing chordal completion that are \emph{inclusion-wise minimal}, 
has been intensively studied, either to find an optimal tree decomposition with an exponential
time algorithm or to find one tree decomposition in polynomial time.
Over the past decades, numerous polynomial algorithms have been provided to compute one minimal chordal completion 
\cite{berry2004, todinca-bouchitte1998, todinca1999, heggernes2006, ohtsuki1976, parra-scheffler1997, rose-tarjan1976}.

However, from an application point of view, computing only one tree decomposition 
might not be satisfactory.
For that reason, Carmeli \emph{et al.} \cite{CarmeliKK17,carmeli2020}
considered the problem of listing all the minimal chordal completions 
of a graph, hence obtaining all the minimal tree decompositions of a graph. 
In the same line of research, Ravid \emph{et al.} \cite{ravid2019} considered the same 
problem by adding a requirement on the order in which solutions are produced.
%

\bigskip
An enumeration algorithm lists every solution of a given problem exactly once.
Since the considered problem can have exponentially (in the input size)
many solutions, the traditional complexity measures are no longer relevant.
Instead, the common approach called \emph{output sensitive} analysis is to bound the time
complexity by a function of the input and the output size.
Johnson \textit{et al.} adapted  in \cite{johnson-yannakakis1988}
the notion of polynomial time algorithm for enumeration algorithms.
An algorithm is said to be \emph{output polynomial} if its complexity can be bounded 
by a polynomial function expressed in the size of both the input and the output. 
%
%
As the number of solutions might be huge, this notion is not fine enough 
to capture the efficiency of the algorithm.
For that reason, Johnson \textit{et al.} further refined this notion by introducing
the \emph{incremental polynomial time}, meaning that the time used to produce a new 
solution of the problem is bounded by a polynomial in the size of the input 
and the size of all the already produced solutions.
They strengthened the concept with the notion of \emph{polynomial delay}: 
the time between the generation of two consecutive solutions 
is bounded by a polynomial in the size of the input only.
The total running time of a polynomial delay algorithm depends only linearly on the size of the output, and since all solutions have to be outputted, this dependency is optimal.

For the enumeration of minimal triangulations of a graph, 
Carmeli \textit{et al.} in \cite{CarmeliKK17,carmeli2020} presented
a highly non-trivial algorithm.
Their algorithm runs in \emph{incremental polynomial time}; its total complexity is quadratic in the number of solutions, and requires exponential space.
The algorithm is based on a result by Parra \& Scheffler \cite{parra-scheffler1997}, according 
to which there is a bijection between the minimal triangulations of a graph and 
the maximal independent sets of a special graph of minimal separators of the 
input graph.
In \cite{carmeli2020}, they proved under \emph{Exponential Time Hypothesis}
that their approach cannot be improved to achieve polynomial delay.
This intractability result also holds for the exponential space.
In \cite{CarmeliKK17,carmeli2020} and at a Dagstuhl seminar \cite{BorosKPS19}, it was left as an open problem   whether a polynomial delay algorithm for this problem could be obtained.
We answer this question by the affirmative.
Our main result is the following theorem.

\begin{theorem}
\label{thm:main}
The minimal chordal completions of a graph can be computed in polynomial 
delay and polynomial space.
\end{theorem}

In addition, our algorithm is simple and can be easily implemented.
Since it is a polynomial delay algorithm, its total complexity  is linear in the number of solutions.
It contrasts with the quadratic one of \cite{CarmeliKK17,carmeli2020}.

\bigskip
The paper is organised as follows. 
In Section \ref{sec:preliminary} we recall the results and techniques used 
in algorithmic enumeration.
In Section \ref{sec:poly-delay}, we present the concepts that will be used for minimal chordal 
completions, and we present a first polynomial delay and exponential space algorithm.
Then in Section \ref{sec:poly-space} we present our main result, namely a polynomial delay and 
polynomial space algorithm to list all the minimal chordal completions of a graph without duplication. 
Finally, in Section \ref{sec:general} we formalize the framework used in Section 
\ref{sec:poly-space} by introducing a new general method called \emph{canonical path 
reconstruction} to design polynomial space enumeration algorithms.

\section{Definitions and preliminary results}
\label{sec:preliminary}

Throughout the article, standard graph theory notations will be used.
A graph, always assumed to be simple, finite and undirected, is denoted by $G = (V,E)$ where $V$ is the set of vertices and $E$ is the set of edges.
When the context is ambiguous, notations $V(G)$ and $E(G)$ can be used.
For any $k\geq 3$, $C_k$ denotes a cycle of length $k$.

We denote by $E^c$ the set of \emph{non-edges} of $G$, that is, the complement of the set $E$ in the larger set of all two-element subsets of $V$.
Then, for a subset $F\subset E^c$ we denote by $\bar{F}$ the set $E^c\setminus F$.

\bigskip
A graph is \emph{chordal}, or \emph{triangulated}, if it does not contain any induced cycle of length $4$ or more.
In other words, every cycle of length more than $3$ of a chordal graph has at least one chord.

Given a graph $G$ and a supergraph $H$ of $G$ on the same vertex set, $H$ is called a \emph{triangulation} of $G$ if $H$ is chordal, and the edges of $H$ which are not edges of $G$ are called \emph{fill edges}.
%
In the whole paper, triangulations, or chordal completions, will be identified with the set of fill edges they induce.
This is why for a graph $G=(V,E)$, we call a set $F\subseteq E^c$ a \emph{chordal completion of $G$} if the supergraph $G_F=(V,E\cup F)$ is chordal.
Let us denote by $\cF$ the set of all chordal completions of $G$ and by $\Fmin$ the set of minimal chordal completions of $G$, with respect to inclusion.
A characterisation of minimal chordal completions is given in \cite[Theorem~2]{rose-tarjan1976}: a chordal completion $F$ of a graph $G$ is minimal \emph{if and only if} for any $f\in F$, the graph $G_F \setminus \{f\}$ has an induced $C_4$.

From now on, $G=(V,E)$ is considered to be an arbitrary input graph that has no particular property and is therefore not assumed to be chordal.
In the whole article, notation $n$ is used for the number of vertices of $G$, that is, $n = \vert V\vert$.
Finally, as our goal is to enumerate all minimal chordal completions of $G$, we may simply refer to them as ``minimal completions''.


The enumeration of minimal chordal completion takes place in the more general task of enumerating the minimal or the maximal subsets of a set system.
A \emph{set system} is a couple $(\cU,\cF)$ where $\cU$ is called the ground set and $\cF\subseteq 2^{\cU}$ is a family of subsets of $\cU$. For a set system $(\cU,\cF)$
we denote by $\Fmin$ (resp. $\Fmax)$ the inclusion-wise minimal (resp. maximal) sets in $\cF$. 
Many enumeration problems consist in enumerating the set $\Fmax$ or $\Fmin$ of a set system $(\cU,\cF)$.
Of course the set $\cF$ is usually not part of the input and we simply assume that a polynomially computable oracle membership is given, \textit{i.e.} one can check whether a subset $F\subseteq \cU$ belongs to $\cF$ in time polynomial in $\cU$.

In \cite{conte-uno2019}, the authors describe a method called \emph{Proximity Search} (by canonical reconstruction) to design polynomial delay algorithms to enumerate $\Fmax$ of a set system $(\cU,\cF)$.
Given a set family $\cF$ on a ground set $\cU$, an \emph{ordering scheme} $\pi$ is a function which associates for every $F\in \cF$ a permutation $\pi(F)=f_1,...,f_{|F|}$ of the elements of $F$ such that for all $i<|F|$, the $i$th prefix $\{f_1,...,f_i\}$ of $\pi(F)$ is in $\cF$.
Notice that a set system $(\cU,\cF)$ has an ordering scheme if and only if  for every $F\in \cF$, there exists $f\in F$ such that $F\setminus \{f\} \in \cF$.
Set systems having this property are called \emph{accessible}.

The method described in \cite{conte-uno2019} is based on the \emph{proximity} between 2 solutions.
While this notion has been defined in a very general context, most of its use cases are based on an ordering scheme.
Given an ordering scheme $\pi$, and given $F_1,F_2\in \cF$ with $\pi(F_2)=f_1,...,f_{|F_2|}$, the \emph{$\pi$-proximity} between $F_1$ and $F_2$, denoted by $\prox{F_1}{F_2}$, is the largest $i\leq |F_2|$ such that $\{f_1,...,f_i\}\subseteq F_1$.
It is worth noticing that the proximity relation between two solutions is not necessarily symmetric.

A polynomial-time computable function $\neig: \Fmax \to 2^{\Fmax}$ is called \emph{$\pi$-proximity searchable}, if for every $F_1, F_2 \in \Fmax$ there exists $F' \in \neig(F_1)$ such that $\prox{F'}{F_2}>\prox{F_1}{F_2}$.
Finally, we say by extension that an ordering scheme $\pi$ of a set system $(\cU,\cF)$ is \emph{proximity searchable} if there exists a  $\pi$-proximity searchable function $\neig$.


One of the major results of \cite{conte-uno2019} is the following theorem.
\begin{theorem}[\cite{conte-uno2019}]\label{thm:prox-searchable}
      Let $(\cU,\cF)$ be a set system and assume that one can find in polynomial time a maximal set $F_0\in \Fmax$.
      If $\cF$ has a proximity searchable ordering scheme, then $\Fmax$ can be enumerated with polynomial delay.
\end{theorem}

The proof of this theorem is based on the analysis of the \emph{supergraph of solutions}.
The supergraph  of solutions is the directed graph having $\Fmax$ as vertex set and there is an arc between $F_1$ and $F_2$ if $F_2\in\neig(F_1)$ (where \neig is the $\pi$-proximity searchable function).
The proximity searchability of the ordering scheme implies that this supergraph of solutions is strongly connected.
Then, the polynomial delay algorithm consists in starting from an arbitrary solution $F_0\in \Fmax$ and performing a traversal of the supergraph of solutions, following at each step the arcs computed on the fly by Function $\neig$.
The strong connectivity of the supergraph of solutions ensures that all solutions 
will be found.
However, with this method, one needs to store in memory the solutions already visited, which in general results in the need of an exponential space.

One of the classical methods in enumeration to avoid the storage of the already outputted solutions is to define a parent-child relation over the set of solutions.
It consists in associating to each $F\in \Fmax$ a parent solution $\parent(F) \in \Fmax$ such that $F\in \neig(\parent(F))$.
Given the $\parent$ function, one can finally define the $\children$ as $\children(F):=\{F' \in \neig{F} \mid \parent(F')=F\}$.
The goal is to define the function $\parent$ in such a way that the arcs of the supergraph of solutions defined by Function $\children$ form a spanning arborescence of the supergraph of solutions, because in such a situation, one does not need to store the already visited solutions in a traversal of the supergraph of solutions.
This method and the way to traverse the spanning arborescence of the supergraph is called \emph{Reverse Search}\cite{avis-fukuda1996} and has been used in many contexts.

Most applications of Reverse Search use a classical parent-child relation originally introduced in \cite{tsukiyama_new_1977} for the enumeration of maximal independent sets of a graph.
This specific parent-child relation has been used in general enumeration frameworks to obtain polynomial delay and polynomial space algorithms \cite{Lawler1980GeneratingAM,cohen_generating_2008}.
Unfortunately, this method can only be used for hereditary set systems and it is not compatible with Proximity Search in general.

In \cite{conte2019} the authors show how to adapt this parent-child relation to commutable set systems (a class that strictly contains hereditary set systems) and in \cite{conte_proximity_2020}, it has been shown that the same method can be combined with Proximity Search to obtain polynomial delay and polynomial space algorithms for commutable set systems.

A set system is \emph{commutable} if for any $X,Y\in \cF$ with $X\subseteq Y$, the following two conditions hold:
\begin{itemize}
    \item {\bf Strong accessibility:}\\
    There exists $f\in Y\setminus X$ such that $X\cup\{f\} \in \cF$
    \item {\bf Commutability:}\\
    For any $a,b\in Y\setminus X$, if $X\cup \{a\} \in \cF$ and $X\cup \{b\} \in \cF$ then $X\cup \{a,b\} \in \cF$
\end{itemize}

More formally, the authors introduce the notion of \emph{prefix-closed} ordering schemes (\textit{cf.} Section \ref{sec:general} for a definition) and they show the following Theorem.

\begin{theorem}[\cite{conte_proximity_2020}]\label{thm:commutable}
Let $(\cU,\cF)$ be a commutable set system and assume that one can find in polynomial time a maximal set $F_0\in \Fmax$.
If $\cF$ has a proximity searchable prefix-closed ordering scheme, then $\Fmax$ can be enumerated with polynomial delay and polynomial space.
\end{theorem}

The polynomial space complexity comes from the definition of a parent-child relation that strongly relies on the commutability property.
In the current paper, we introduce a new way of defining parent-child relations called \emph{canonical path reconstruction} that does not require the commutability property.
Since the set of chordal completions is not a commutable set system, this new method will be used in Section \ref{sec:poly-space} to obtain a polynomial delay algorithm.
Then, in Section \ref{sec:general} we improve Theorem \ref{thm:commutable} by proving the following which applies to non-commutable set systems.

\begin{theorem}\label{thm:main_general}
Let $(\cU,\cF)$ be a set system and assume that one can find in polynomial time a maximal set $F_0\in \Fmax$.
If $\cF$ has a proximity searchable prefix-closed ordering scheme, then $\Fmax$ can be enumerated with polynomial delay and polynomial space.
\end{theorem}

\section{Enumeration of minimal chordal completions: polynomial delay}\label{sec:poly-delay}
In this section we present a polynomial delay and exponential space algorithm
to list all minimal chordal completions of a graph. 
In Section \ref{sec:poly-delay:def} we present all the concepts necessary 
to define an ordering scheme that will suit to minimal chordal completions.
Then in Section \ref{sec:poly-delay:algo}, we present the neighbouring function and 
prove that this function is proximity searchable.
As a consequence, together with Theorem \ref{thm:prox-searchable}, we obtain a polynomial delay algorithm.

\subsection{Ordering scheme for chordal graphs}
\label{sec:poly-delay:def}

%
A class of graphs $\cC$ is called \emph{sandwich-monotone} \cite{heggernes2009} if for any two graphs $H_1$ and $H_2$ in $\cC$ such that $E(H_1)\subsetneq E(H_2)$, there exists an edge $e\in E(H_2)\setminus E(H_1)$ such that $H_2\setminus \{e\}$ is also in $\cC$.
This property is equivalent to being \emph{strongly accessible} in terms of set systems (see for example \cite{conte2019}).

In \cite[Lemma~2]{rose-tarjan1976}, the authors prove that the class of chordal graphs is sandwich-monotone.
This immediately implies that if $F_1\subsetneq F_2$ are two chordal completions of the graph $G$, there exists $e\in F_2\setminus F_1$ such that $F_2\setminus \{e\}$ is a chordal completion, or equivalently there exists $e\in F_2\setminus F_1$ such that $F_1\cup \{e\}$ is a chordal completion.
That rephrases as the following Lemma.

\begin{lemma}\label{lem:sandwich-monotone}
Chordal graphs are sandwich-monotone.
In particular, chordal completions of $G$ are sandwich-monotone.
\end{lemma}

The general idea is then to use the sandwich-monotonicity of chordal completions to derive a suitable ordering scheme.
The Proximity Search framework introduced in \cite{conte2019} has been designed to enumerate inclusion-wise maximal subsets of a set system.
However, the same framework could be applied to the enumeration of inclusion-wise minimal subsets of a set system, by considering the complements of the solutions.
To this effect, we will not work directly with the edge sets of the completions but rather with the sets of their \emph{non-edges}.
That is to say, we will consider for any completion $F$ its complement $\bar{F}:=  E^c \setminus F$.
The goal of this section is then to apply Proximity Search to the set $\bar{\cF}:=\{\bar{F} \mid F\in \cF\}$ in order to enumerate the set $\bar{\cF}_{\text{max}}:=\max(\bar{\cF})=\{\bar{F} \mid F\in \cF_{min}\}$.


\bigskip
From now on, the elements of $E^c$ (that is to say, the non-edges of $G$) are assumed to be arbitrarily ordered.
The following definitions are similar to the ones given in \cite{conte2019} but are defined on $\cF$ instead of $\bar{\cF}$.
Let $F$ be a chordal completion of $G$ and $X$ be any subset of $E^c$.
Consider the set $X$ as the set from which we are allowed to remove elements.
We define:
\begin{itemize}
    \item $\candidates(F,X):=\{e\in X\cap F : F\setminus \{e\} \in \cF\}$
    \item $\candidates(F):=\candidates(F,F)$
    \item $\cc(F,X):=\min(\candidates(F,X))$
    \item $\cc(F):=\min(\candidates(F))$ 
\end{itemize}

Now, given a chordal completion $F$ and a set $X\subseteq E^c$, we denote by $\del(F,X)$ the chordal completion included in $F$ by iteratively removing $\cc(F,X)$ from $F$ at each step.
Finally, we define $\del(F):=\del(F,F)$.
Note that, for any $F$, $X$, computing $\del(F,X)$ corresponds to the following procedure.


\begin{function}
\While{$\candidates(F,X)\neq \emptyset$}{
    remove $\cc(F,X)$ from $F$\;
  }
\Return{$F$}
\label{alg:DEL}
\caption{Del($F,X$)}
\end{function}
%

\begin{rmq}\label{rmq:del_min}
By Lemma~\ref{lem:sandwich-monotone}, if $F\in \cF$ then $\del(F)\in \Fmin$.
That is to say, $\del$ can be used to turn a chordal completion into a minimal one in a canonical way.
\end{rmq}

Also, as $E^c$ is a chordal completion (it corresponds to the clique completion) of $G$, the next Lemma holds.

\begin{lemma}
	If $F$ is a minimal chordal completion of $G$, then $\del(E^c,\bar{F})=F$.
\end{lemma}

\begin{proof}
By Remark~\ref{rmq:del_min}, it holds $\del(E^c,\bar{F})\in \Fmin$, and $F\subset \del(E^c,\bar{F})$.
As a result, since the only minimal solution containing $F$ is $F$ itself, then $\del(E^c,\bar{F})=F$.
\end{proof}

The procedure followed to compute $\del(E^c,\bar{F})$ provides an ordering on the elements of $\bar{F}$ by considering the order in which elements of $\bar{F}$ are removed to obtain $F$.
From this ordering, we can define for any chordal completion $F$ the canonical ordering $\can(\bar{F}):=s_1,...,s_{|\bar{F}|}$ of $\bar{F}$ as follows:
\begin{itemize}
      \item $s_1 := \cc(E^c,\bar{F})$;
      \item for all $1\leq i < \vert\bar{F}\vert$, $s_{i+1} := \cc(E^c \setminus \{s_1,...,s_i\},\bar{F})$.
\end{itemize}

For a set $\bar{F}$  and its canonical ordering $\can{(\bar{F})}:=s_1,...,s_{|\bar{F}|}$ we define $\bar{F}^i$ as the set of elements $\{s_1,...,s_i\}$ of $\bar{F}$ and we define  $F^{i}:=E^{c}\setminus \bar{F}^{i}$.
By definition of $\candidates(E^c,\bar{F})$, any prefix $\bar{F}^{i}$ of this ordering belongs to $\bar{\cF}$.
In other words, $F^{i}$ denotes the chordal completion of $G$, not necessarily minimal, obtained by removing the $i$th prefix of $\can(\bar{F})$ from the clique completion.
Thus, the canonical ordering $\can$ is an ordering scheme of $\bar{\cF}$.

This ordering will be used to measure the proximity between two solutions in order to call the Proximity Search algorithm for minimal chordal completions.
Note that it can be computed in polynomial time for any solution, the complexity being essentially this of calling Function $\del$.

\bigskip
We now define the notion of proximity between two solutions that will be used in the sequel.
For $F_1$ and $F_2$ two minimal completions of $G$, let $\can(\bar{F_2})=f_1,...,f_k$ be the canonical ordering of $\bar{F_2}$.
The \emph{proximity} $\prox{\bar{F_1}}{\bar{F_2}}$ between $\bar{F_1}$ and $\bar{F_2}$ is defined as the largest $i\leq k$ such that $\{f_1,...,f_i\} \subseteq \bar{F_1}$.
It is worth noticing that the proximity relation between two solutions is not necessarily symmetric.
The notion defined here corresponds to the ``standard'' proximity measure that is usually adopted when using the Proximity Search framework \cite{conte-uno2019, conte2019, brosse2020}.
Note that the proximity is defined here on the set of elements of $E^c$ (that is to say, non-edges of $G$) which do $\emph{not}$ belong to the chordal completion.
Since we are working on both the completions and their complement, by a slight abuse of language, when we speak about 
the proximity between two minimal completions, we actually mean the proximity between their complement sets. 

%
%
%
%
\subsection{Polynomial delay algorithm}
\label{sec:poly-delay:algo}

We show in this section that the ordering scheme $\can$ defined above is proximity searchable.
Since the idea is to consider the complement sets of minimal completions rather than the completions themselves, we will seek to maximise the set of common \emph{non-edges} between two minimal completions in order to increase the proximity.
So, we actually show that $\can$ is a proximity searchable ordering scheme of $\bar{\cF}$.

The first goal is to define a suitable neighbouring function on the class of chordal graphs.
Any solution must have a polynomial number of neighbours, each of them being computable in polynomial time in order to guarantee the polynomial delay when applying Proximity Search.

\bigskip
%
Given a chordal graph $H$ and an edge $e=\{x,y\}\in E(H)$, the \emph{flip operation} $\flip(H,e)$ consists in removing $e$ from $H$, and turning the common neighbourhood of $x$ and $y$ into a clique.
More formally, if $H' =  \flip(H,e)$, then $V(H') = V(H)$ and $E(H'):= (E(H) \setminus \{e\}) \cup \{uv \mid u,v \in N_{H}(x) \cap N_{H}(y)\} $.
The flip operation is illustrated in Figure~\ref{flip}.

\begin{figure}[!ht]
	\centering
	\includegraphics[width=0.45\textwidth]{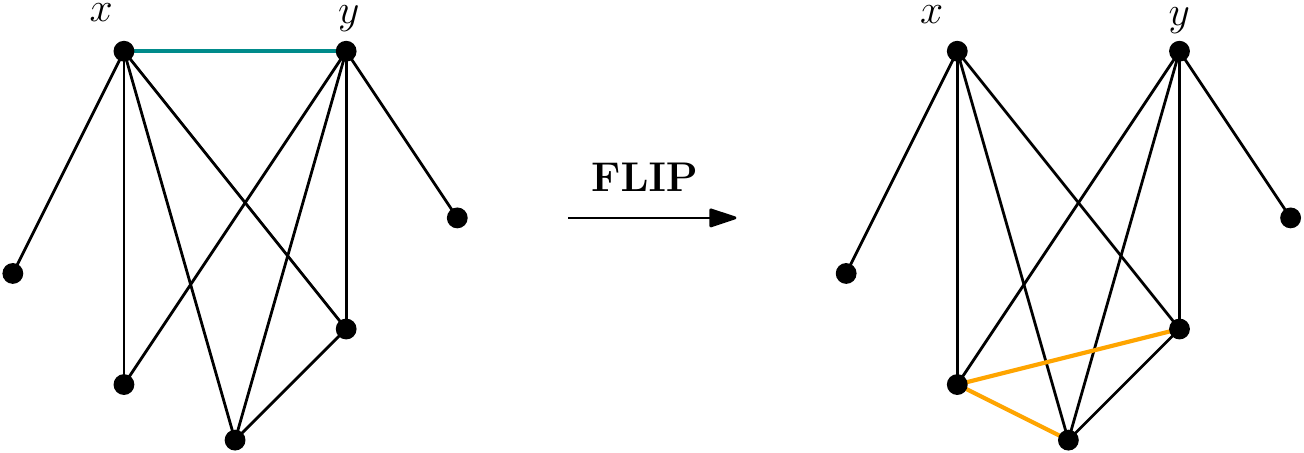}
	\caption{The flip operation on $e = \{x,y\}$. The common neighbourhood of $x$ and $y$ is turned into a clique.}
	\label{flip}
\end{figure}

Since $H$ is chordal, the removal of $e$ can create several chordless $C_4$ of which $e$ was the only chord.
We will see that completing the common neighbourhood of $x$ and $y$ into a clique adds the missing chords to all these $C_4$.
In other words, the flip operation preserves the class of chordal graphs, as stated next.

\begin{lemma}\label{lem:flip_preserve_cordaux}
	Let $H$ be a chordal graph, and $e=\{x,y\}$ be an edge of $H$.
	Then the graph $H':=\flip(H,e)$ is also chordal.
\end{lemma}

\begin{proof}
Assume (for contradiction) that $H'$ contains a chordless cycle $C$ of length $\ell >3$.

Suppose $C$ contains only edges of $H$.
Since $H$ is chordal, it does not contain long induced cycles.
Then necessarily $C$ is a $C_4$ of $H$ of which $e$ was the unique chord, otherwise there would be an induced cycle of length at least $4$ (either $C$ or a cycle made from $e$ and edges of $C$) in the chordal graph $H$.
By definition, the flip operation adds an edge between the other two vertices of $C$, meaning that $C$ has a chord in $H'$.
This cannot happen.

Therefore, $C$ contains an edge $e'=\{z,t\}$ added by the flip operation (\textit{i.e.} $e'$ is an edge of $H'$ but not of $H$).
Since $C$ contains the edge $e'$, it does not contain any other fill edge added by the flip operation, otherwise a chord of $C$ would also be added by the flip.
Plus, we assumed that $e'$ is added by the flip operation, so necessarily both $z$ and $t$ belong to $N_H(x)\cap N_H(y)$.
Therefore, $P := C - \{e'\}$ is a $z-t$ path of $H'$, disjoint from $xy$.
Remark that $P$ is an induced path of $H$ since it contains only edges of $H$, that is, $z$ and $t$ are in the same connected component of $H\setminus\{x,y\}$.

\begin{claim}
	Every vertex of $P$ is in $N_H(x)\cap N_H(y)$.
\end{claim}

\begin{proofclaim}
	Suppose that there exists a vertex $s$ of $P$ which is not a neighbour of $x$.
	We will show that there exists in $H$ a chordless cycle containing $s$ and $x$.

	As $s$ is not a neighbour of $x$, and $z$ is by hypothesis, there exists $q\in N(x)$ a vertex of $P$ that is between $z$ and $s$, and closest to $s$ on $P$.
	Similarly, let $r\in N(x)$ be a vertex of $P$ that is between $t$ and $s$, and closest to $s$.

	The $q-r$ subpath of $P$ induced by all vertices between $q$ and $r$ is therefore a chordless path of $H$ of which all internal nodes are non-neighbours of $x$.
	It follows that adding $x$ to this $q-r$ subpath creates a chordless cycle of length at least $4$ in $H$.
	This is excluded, so all vertices of $P$ are neighbours of $x$ in $H$.

	By symmetry in $x$ and $y$, we also deduce that all vertices of $P$ are neighbours of $y$.

	The claim is proved: every vertex of $P$ is in $N_H(x)\cap N_H(y)$.
\end{proofclaim}

By the previous Claim, the flip operation turns the cycle $C$ into a clique.
Hence $C$ has length $3$, a contradiction.
\end{proof}

When dealing with chordal completions, the notation of the flip operation will be slightly adapted: for $F\in\cF$, and $e\in F$ ($e$ is always chosen among the fill edges of the completion), we write $\flip(F,e)$ instead of $\flip(G_F,e)$.

From Lemma~\ref{lem:flip_preserve_cordaux}, we are then able to deduce the following.

\begin{lemma}\label{lem:flip}
	Let $F$ be a chordal completion of a graph $G$.
    If $F\in \Fmin$, and $e\in F$, then $\flip(F,e)\in \cF$.
\end{lemma}

\begin{proof}
	The only edge that is in $F$ but not in $\flip(F,e)$ is $e$.
	Since $e\in F$, it implies that all edges of $G$ are still in $G_{\flip(F,e)}$.
	As $\flip(F,e)$ is chordal by Lemma~\ref{lem:flip_preserve_cordaux}, in particular it is a chordal completion of $G$.
\end{proof}

Observe that $\flip(F,e)$ is a chordal completion of $G$ that is not necessarily minimal.

We are now ready to explain the neighbouring function used to enumerate $\Fmin$.
Given $F\in \Fmin$ and $e\in F$, we define the successor of $F$ according to $e$ as the minimal completion $\successor(F,e):=\del(\flip(F,e))$.
We now define the neighbours of a solution $F$ as $\neig(F):=\{\successor(F,f) | f\in F\}$.
As a corollary of Remark~\ref{rmq:del_min} and Lemma~\ref{lem:flip}, it is true that $\successor(F,e)\in \Fmin$. Observe also that each solution has a polynomial number of neighbours since $\neig(F)\leq |F|$ .



\bigskip
Now the neighbouring function is  properly defined, it is easy to notice that 
it can be computed in polynomial time. It remains to prove that the 
function $\neig$ is $\can$-proximity searchable.
This will allow us to use Proximity Search on the set of minimal chordal completions of a graph $G$.

\begin{lemma}\label{lem:proximity}
   Let $F_1$ and $F_2$ be two minimal chordal completions of a graph $G$.
   Let $f_1,\ldots,f_k = \can(\bar{F_2})$ and let $i:=\prox{\bar{F_1}}{\bar{F_2}}$.
   Then the following statements hold:
   \begin{enumerate}
       \item $f_{i+1} \notin \bar{F_1}$; \label{item:proximity-1}
       \item $\prox{\overline{\successor(F_1,f_{i+1})}}{\bar{F_2}}>i$. \label{item:proximity-2}
   \end{enumerate}
\end{lemma}

\begin{proof}
\ref{item:proximity-1}.
By definition of $i$ as the length of the largest prefix of $\can(\bar{F_2})$ included in  $\bar{F_1}$, it holds $f_{i+1} \notin \bar{F_1}$.

\ref{item:proximity-2}.
Let $F':=\flip(F_1,f_{i+1})$, we will show that $\{f_1,...,f_{i+1}\} \subseteq \bar{F'}$.
Since $\successor(F_1,f_{i+1}) \subseteq F'$, it implies that $\{f_1,...,f_{i+1}\} \subseteq \overline{\successor(F_1,f_{i+1})}$, and finally $\prox{\overline{\successor(F_1,f_{i+1})}}{F_2} \geq i+1$.

First, notice that by definition of $\flip$, $f_{i+1} \notin F'$.
Let $x$ and $y$ be the two endpoints of $f_{i+1}$.
Assume for contradiction that an edge $f_j$, $j\leq i$ is added when completing $N_{G_{F_1}}(x)\cap N_{G_{F_1}}(y)$ into a clique.
As stated earlier, the notation $\bar{F_2}^{i+1}$ is used to denote the set $\{f_1,...,f_{i+1}\}$, which is included in $\bar{F_2}$, and $F_2^{i+1}$ represents the associated chordal completion.

Let $u$ and $v$ be the two endpoints of $f_j$.
Then $x,u,y,v$ forms a $C_4$ of which $f_{i+1}$ was the unique chord in $G_{F_1}$.
By definition of $\can(\bar{F_2})$, $F_2^{i+1}$ is a chordal completion of $G$, and since neither the chords $f_j$ nor $f_{i+1}$ belong to it, at least one of the edges $xu$, $uy$, $yv$ or $vx$ does not belong to $F_2^{i+1}$ since otherwise $x,u,y,v$ would form a chordless $C_4$ in $F_2^{i+1}$.
Assume without loss of generality that $xu \notin F_2^{i+1}$.
Since $\bar{F_2}^{i+1} = \{f_1,...,f_{i+1}\}$, there exists $\ell\leq i$ such that $xu=f_\ell$.
But then we have found an edge $f_\ell \notin \bar{F_1}$ with $\ell \leq i$, which contradicts the assumption $\prox{\bar{F_1}}{\bar{F_2}}=i$.

Consequently, $\prox{\overline{\successor(F_1,f_{i+1})}}{\bar{F_2}}>i$.
\end{proof}

As a consequence, since the polynomial computable function $\successor$ is able to increase the proximity between 2 solutions, it proves that the ordering scheme defined by $\can$ is proximity searchable.
More formally, the function $\neig(\bar{F}):=\{\bar{X} \mid X\in \neig(F)\}$ is a $\can$-proximity searchable function and $\can$ is then a proximity searchable ordering scheme for $\bar{\cF}_{max}$.

Therefore, by Theorem \ref{thm:prox-searchable}, $\bar{\cF}_{max}$ or equivalently $\Fmin$ can be enumerated with polynomial delay.
This is summarized in the following Theorem.

\begin{theorem}
	There exists a polynomial delay algorithm for the enumeration of minimal chordal completions of a graph.
\end{theorem}

To prove that $\can$ is proximity searchable, we only used the fact that it is an ordering scheme of $\bar{\cF}$.
Thus any ordering scheme of $\bar{\cF}$ is actually proximity searchable.



%
%
%

Yet, applying Proximity Search usually gives polynomial delay with exponential space: to avoid duplication it is necessary to store all the generated solutions in a lookup table.
Each newly generated solution is then searched in the table in polynomial time.

\section{Polynomial space}\label{sec:poly-space}

We would like to make the enumeration process work in polynomial space, since this could lead to an algorithm that is more usable in practice.
As stated in Theorem~\ref{thm:commutable}, it is proven in \cite{conte2019} that if the ordering scheme is prefix-closed and $\cF$ is a \emph{commutable}
set system, one can design a polynomial delay and polynomial space algorithm for the enumeration of $\Fmax$.

An ordering scheme $\pi$ over $\cF$ is called \emph{prefix-closed}, if for all $F\in \cF$, there exists an ordering $<_F$ of the elements of $\candidates(F)$ such that $\pi(\bar{F})=f_1,...,f_\ell$ \emph{if and only if} for any $i< \ell$, it holds $f_{i+1}=\min\limits_{<_{F^i}}(\candidates(F^i)\cap \bar{F})$.
We know by definition of $\can$ that this ordering is prefix-closed (with the ordering of $E^c$).
By Section~\ref{sec:poly-delay}, $\can$ is also proximity searchable.
As stated earlier, chordal completions form a strongly accessible set system since they are sandwich-monotone, and so are their complements.
However, the set of chordal completions is not a commutable set system, and neither is the set of their complements.


\bigskip

This section is devoted to the definition of a suitable parent-child relation over $\Fmin$ which does not require the system to be commutable, in order to obtain a polynomial space algorithm.
In the supergraph of solutions, we identify a reference solution named $F_0$ and we manage to relate any other solution $F$ to $F_0$.
To uniquely determine the parent of a solution, we introduce a new concept called 
\emph{canonical path reconstruction}.
The idea is, for any solution $F$ distinct from $F_0$, to identify a path in the supergraph of solutions from $F_0$ to $F$ in a unique manner.
Then the parent of solution $F$ is defined as the immediate predecessor of $F$ on this path.
In addition, we are able to compute this path in polynomial time for any solution. 

To ensure that such a path can be uniquely determined and computed in polynomial time
we rely on the specific structure of chordal completions and more specifically 
on the fact that the canonical ordering we defined on chordal completions is prefix-closed.
To the best of our knowledge, it is the first algorithm to define a parent-child 
relation in this manner.

The parent-child relation will define a spanning arborescence of the supergraph of solutions rooted at $F_0$.
This way, the enumeration algorithm sums up to perform a traversal of the tree in a depth-first manner.
%

It is somehow counter-intuitive to observe that the canonical path of a solution is not necessarily part of the final arborescence rooted at $F_0$.
This canonical path is only computed to find the last solution in the path before $F$ to define it as the parent of $F$.
We will prove that while the so defined parent-child relation does not exactly follow the canonical paths, 
it still forms a spanning arborescence of the solution set rooted at $F_0$.

\bigskip
Let us consider the total ordering $\prec$ over $\bar{\cF}_{min}$ defined as $\bar{F_1}\prec \bar{F_2}$ if $\can(\bar{F_1})$ is lexicographically smaller than $\can(\bar{F_2})$.

\begin{lemma}\label{lem:prefix}
     Let $F_1,F_2\in \Fmin$ with $\can(\bar{F_1})=f_1,...,f_\ell$ and  $\can(\bar{F_2})=t_1,...,t_k$.
     Assume furthermore that there exists $j\leq \ell$ such that $\{f_1,...,f_j\} \subset \bar{F_2}$.
     Then one of the two possibilities holds:
     \begin{enumerate}
          \item $\bar{F_2} \prec \bar{F_1}$;\label{item:prefix-1}
          \item $f_i=t_i$ for all $i\leq j$.\label{item:prefix-2}
      \end{enumerate}
\end{lemma}

\begin{proof}
Let $i^*$ be the smallest index such that $t_{i^*}\neq f_{i^*}$.
If $i^*>j$, then \ref{item:prefix-2} holds.
Else, $i^*\leq j$.
In this case, we deduce from the minimality of $i^*$ that $$\bar{F_1}^{{i^*}-1} = \{f_1,...,f_{{i^*}-1}\} = \{t_1,...,t_{i^*-1}\} = \bar{F_2}^{i^*-1}~.$$
Since $f_{i^*} \in \bar{F_2}$ and since $f_{i^*}\in \candidates(F_1^{i^*-1},\bar{F_1})$ by definition of $\can(\bar{F_1})$, we deduce $f_{i^*} \in \candidates(F_2^{i^*-1},\bar{F_2})$.
The canonical ordering $\can$ is prefix-closed, therefore we know that $t_{i^*}=\min(\candidates(F_2^{i^*-1},\bar{F_2}))$.
We deduce from it that $t_{i^*} \leq f_{i^*}$ and since $t_{i^*}\neq f_{i^*}$, we have $t_{i^*}<f_{i^*}$ which proves $\bar{F_2}\prec \bar{F_1}$.
\end{proof}

The algorithm starts by computing $F_0 := \del(E^c)$ in polynomial time.
This solution $F_0$ will be used as the reference solution.
Note that we could start from any solution and the results would still be valid, 
but for simplicity we start from a solution that can easily be identified.

For a minimal chordal completion $F$ with canonical ordering $\can(\bar{F}):=f_1,...,f_\ell$, we define the \emph{canonical path of $F$} as the sequence of minimal completions
$F_0,...,F_k=F$ starting at $F_0$ such that for all 
$1\leq i\leq k$, $F_i = \successor(F_{i-1},f_{(\prox{\bar{F}_{i-1}}{\bar{F}})+1})$.
Remark that the edge $f_{(\prox{\bar{F}_{i-1}}{\bar{F}})+1}$ is part of  $F_{i-1}$ and the call to $\successor$ is correct. 
Let us also note that by construction, the proximity strictly increases along the path. 
Then, $\parent(F)$ is defined for any $F \neq F_0$ as the last minimal completion $F_{k-1}$ in the canonical path of $F$.

As remarked before, except $F_{k-1}$, the other ancestors of $F$ may not be part of its canonical path. 
For instance $F_{k-2}$ may not be the parent of $F_{k-1}$. 

The general $\parent$ function can be computed with Algorithm \ref{proc:parent}.


\begin{function}\label{proc:parent}
\caption{Parent($F$)}
 Compute $f_1,\ldots, f_\ell := \can(\bar{F})$ \;
 Compute $F_{current} := \del(E^c)$ \;
 \While{$F_{current} \neq F$}{
    $f := f_{(\prox{\overline{F_{current}}}{\bar{F}})+1}$\;
    $F_{prec} := F_{current}$ \;
    $F_{current} := \successor(F_{prec},f)$ \;
 }
 \Return{$F_{prec}$}
\end{function}


Then we define $\children(F):=\{F'\in\neig(F) \mid \parent(F')= F\}=\{\successor(F,e) \mid e\in F,~\parent(\successor(F,e))= F\}$.


\begin{lemma}\label{lem:atleastj}
Let $F$ be a minimal chordal completion of $G$, and let $F_0,\ldots, F_k=F$ be the canonical path of $F$.
For all $j\leq k$, $\prox{\bar{F_j}}{\bar{F}}\geq j$.
\end{lemma}

\begin{proof}
     By construction of the canonical path of $F$, the proximity strictly increases at each step by at least one.
\end{proof}

By Lemma \ref{lem:atleastj} the length of the canonical path of $F$ is smaller than $|\bar{F}|$. Since the construction of the canonical path of $F$ is done by applying at most $|\bar{F}|$ times Function $\successor$, $\parent(F)$ is computable in polynomial time for any $F\in \Fmin$.

Now, since the computation of $\children(F)$ is done by applying Function $\successor$ followed by Function $\parent$ at most $|F|$ times, it is also computable in polynomial time.

Once the $\children$ function defined, one just has to apply the classical algorithm to visit and output all solutions whose height level description is presented in Algorithm \ref{alg:efficient}.
To prove the correctness of the algorithm, we just have to prove that the function $\children$ defines a spanning arborescence on the set of solutions rooted at $F_0$.


\begin{algorithm}[h!]
	\SetKwInOut{Input}{input}
	\SetKwInOut{Output}{output}
	\Input{A graph $G=(V,E)$}
	\Output{All minimal chordal completions of $G$}
	\SetKwFunction{Fenum}{enum}
	
	\BlankLine
	$F_0 := \del(E^c)$\;
	Call $\text{\textsc{enum}}(F_0)$\;
	
	\BlankLine
	\SetKwProg{Function}{Function}{:}{}
	\Function{\Fenum($F$)}{
		\texttt{/* Output $F$ if recursion depth is even */}\;
		\ForEach{$F'\in \children(F)$}{$\Fenum(F')$\;}
		\texttt{/* Output $F$ if recursion depth is odd */}\;
	}
	
	\caption{Efficient enumeration of minimal chordal completions}
	\label{alg:efficient}
\end{algorithm}


Observe that the naive implementation of Algorithm \ref{alg:efficient} may use exponential space, since the height of the recursion tree might be exponential.
Indeed the recursion tree corresponds to the arborescence defined by the parent-child relation and we are not able to guarantee that its height is polynomial.
However, since the functions $\parent$ and $\children$ are computable in polynomial time we don't need to store the state of each recursion call.
The classical trick introduced in \cite{tsukiyama_new_1977} and formalized in \cite{avis-fukuda1996} as part of the Reverse Search algorithm is to use Function $\parent$ to perform the backtrack operation on the fly.
Indeed this function is able to navigate backward in the tree, removing the need of keeping in memory all recursion calls.
Thus, the algorithm only needs to keep in memory the current solution.

\begin{theorem}\label{lem:children}
     Algorithm~\ref{alg:efficient} outputs all minimal chordal completions of the input graph without duplication, with polynomial delay and using polynomial space.
\end{theorem}

\begin{proof}

Assume that some solutions of $\Fmin$ are not outputted by the algorithm, and let $F$ be the smallest one with respect to $\prec$ with $\can(\bar{F}):=f_1,...,f_\ell$.
Let $F_0,...,F_k=F$ be the canonical path of $F$.
We prove that all $F_i$, $i< k$ will be outputted by the algorithm.
This way, $\children(F_{k-1})$ will be produced, contradicting the fact that $F$ has not been outputted.

Let $i$ be the smallest index such that $F_i$ has not been processed by the algorithm.
By minimality of $\bar{F}$ with respect to $\prec$, we know that $\bar{F}\prec \bar{F_i}$.
Hence by Lemmas~\ref{lem:prefix} and \ref{lem:atleastj}, $f_1,...,f_i$ is a prefix of $\can(\bar{F_i})$.
So the canonical path of $F_i$ is precisely $F_1,...,F_i$ and $\parent(F_i)=F_{i-1}$.
Now, by minimality of $i$, we know that $F_{i-1}$ has been outputted by the algorithm, and $F_i\in \children(F_{i-1})$ will be outputted during the processing of $F_i$.

\bigskip
Polynomial delay and polynomial space complexities follow directly from the previous discussions.
The delay is at most twice the time needed to compute Function $\children$.
\end{proof}




\section{Canonical path reconstruction: a general approach}\label{sec:general}

In fact, the result we proved for chordal completions is part of a more general framework.
In this section, we show how to extend the algorithm obtained for chordal graphs, in order to apply it to other graph classes.



As stated in Theorem~\ref{thm:commutable}, the authors of \cite{conte_proximity_2020} proved that if a proximity searchable ordering scheme is prefix-closed, then one can design a polynomial delay and polynomial space algorithm to enumerate $\Fmax$ whenever $\cF$ is a commutable system.
We prove here that the same remains true even if the system is not commutable.
To prove that, we introduce a new method called \emph{canonical path reconstruction} to define parent-child relations over $\Fmax$ which does not require the commutability condition.

Given $F\in\cF$ we denote by $F^{+}$ the set $\{x\in \cU : F\cup\{x\} \in \cF\}$ of candidates for $F$.
Assuming that an ordering scheme $\pi$ is fixed, we denote by $F^i$ the $i$th prefix of $F$ according to $\pi$.

The notion of prefix-closed ordering has been introduced in \cite{conte2019} and it appears to be a key property to design parent-child relations.
Intuitively assume that for each non maximal element of $F\in\cF$ we have a preference given by a total ordering $<_{F}$ over the potential elements that can be added to $F$ (i.e. a total ordering over $F^{+}$).
This preference depends on the set $F$ and the preference among two elements may vary from one $F$ to another.
Then an ordering scheme $\pi$ is said to be prefix-closed if for each prefix $F_i$ of $F$ the next element $f_{i+1}$ corresponds to the most preferred element that remains in $F$ according to the preference relation $<_{F^i}$.
More formally :

\begin{definition}\label{def:prefix_closed}
    An ordering scheme $\pi$ is \emph{prefix-closed}, if there exists an  ordering $<_F$ of the elements of $F^{+}$ for each  $F\in \cF$ that verify the following condition :  
	For every $F\in \cF$, $\pi(F)=f_1,...,f_\ell$ \emph{if and only if} $f_{i+1}=\min\limits_{<_{F^i}}({F^i}^{+}\cap F)$ for any $i< \ell$.
\end{definition}

The goal of this section is to prove the following Theorem.

\begin{theorem}\label{thm:principal}
    If $\pi$ is a polynomial time computable ordering scheme for $\cF$ which is both proximity searchable and prefix-closed, then $\Fmax$ can be enumerated with polynomial delay and polynomial space.
\end{theorem}

Notice that any strongly accessible set system has a prefix-closed ordering scheme. Similarly as the one shown for chordal completions one can choose an arbitrary ordering on the elements of $\cU$ and simply define $\pi(F):=f_1,...,f_\ell$ with $f_i:=\min(F^{{i}^{+}} \cap F)$. The strong accessibility ensures that this ordering is well defined.

To prove Theorem \ref{thm:principal}, we define the general parent-child relation over $\Fmax$ based on the \emph{canonical path reconstruction method}.
Let then $\pi$ be a prefix-closed, proximity searchable ordering scheme on the set family $\cF$.

\bigskip

Whenever an ordering scheme $\pi$ is proximity searchable, the supergraph of solutions defined by the $\pi$-proximity searchable function $\neig$ is strongly connected.
The goal will be to choose a reference solution $F_0$ and to identify for each other solution $F$ a canonical path from $F_0$ to $F$.
The parent of $F$ will be defined as the last solution of this path.

For a prefix-closed ordering scheme $\pi$, let us define a total ordering $\prec_{\pi}$ over $\Fmax$.
Let $F_1, F_2\in \Fmax$, $\pi(F_1)=t_1,...,t_{|F_1|}$, $\pi(F_2)=f_1,...,f_{|F_2|}$ and let $j\geq 0$ be the largest index such that $F_{1}^{j}=F_{2}^{j}$. 
Let us denote $F:=F_1^{j}=F_2^{j}$.
Then we have $F_1 \prec_{\pi} F_2$ if $t_{j+1} <_{F} f_{j+1}$.

We can now define the what will be the canonical path of a solution $F\in \Fmax$.
To this end, we need to define a function $\Next : \Fmax \times \Fmax \to \Fmax$ such that given $F_1,F_2\in \Fmax$ with $\prox{F_1}{F_2}=i$,
$\Next(F_1,F_2)\in \{F\in \neig(F_1):\prox{F}{F_2}>i\}$.
Since we assumed that $\pi$ is proximity searchable, we know that the set $\{F\in \neig(F_1)\mid \prox{F}{F_2}>i\}$ is not empty.
Hence, to define $\Next(F_1,F_2)$, one just has to choose deterministically an element of the set $\{F\in \neig(F_1):\prox{F}{F_2}>i\}$.
If the function $\neig$ produces solutions in a deterministic order, then $\Next(F_1,F_2)$ can be chosen as the first $F\in \neig(F_1)$ such that $\prox{F}{F_2}>i$.
Otherwise and to be as general as possible, let us define $\Next(F_1,F_2):=\min\limits_{Lex}\{F\in \neig(F_1):\prox{F}{F_2}>i\}$ as the lexicographically smallest set of $\{F\in \neig(F_1):\prox{F}{F_2}>i\}$.

Given a reference solution $F_0$, the \emph{canonical path} of $F\in\Fmax$ is then defined as the sequence $F_0,...,F_k$ of elements of $\cF$ such that:
\begin{itemize}
    \item $F_k=F$
    \item $F_{i+1}=\Next(F_i,F)$ for all $i<k$
\end{itemize}

Each $F\in \Fmax$ has a canonical path.
Indeed, applying Function $\Next$ increases the proximity with $F$, until $F$ is eventually reached.
We define the parent of a solution $F$ of canonical path $F_0,...,F_k$ as $\parent(F):= F_{k-1}$.
The set of children of $F$ is then defined as $\children(F):=\{F'\in \neig(F) \mid \parent(F') = F\}$.

\begin{lemma}\label{lem:atleastj2}
Let $F\in \Fmax$ and let $F_0,...,F_k$ be its canonical path.
Then for all $i<j\leq k$, $\prox{F_j}{F}>\prox{F_i}{F}$.
\end{lemma}

\begin{proof}
     By construction of the canonical path of $F$ the proximity increases at each step by at least $1$.
\end{proof}

The immediate following Corollary ensures that canonical paths are of polynomial size.

\begin{corollary}\label{cor:smallPath}
      If $F\in\Fmax$, then its canonical path is of size at most $|F|+1$.
\end{corollary}

\begin{proposition}\label{prop:}
       The function $\parent$ can be computed in time $O(\cP+f^{2}\cN)$ where $f$ is a maximum size of a solution in $\Fmax$, $\cN$ is the complexity of  function $\neig$ and $\cP$ is the time needed to compute the function $\pi$.
\end{proposition}
\begin{proof}
     To compute the $\parent(F)$ we first need  to compute $\pi(F)$ and then to compute the canonical path of $F$.  By Corollary \ref{cor:smallPath}, the path is of length at most $f$, so we only need to call $f$ times the function $\Next$. To compute $\Next(T,F)$, we compute $\neig(T)$ and for each $F'\in \neig(T)$ we compute the proximity between $F'$ and $F$ and check whether $F'$ is lexicographically smaller than the current best candidate (notice that we don't need to compute $\pi(F')$). Since both operations can be done in $O(f)$, and since $|$\neig(T)$|<\cN$ the function $\Next$ can be computed in time $O(f\cN)$.
\end{proof}

Since the computation of $\children(F)$ is done by computing first the set $\neig(F)$ and then filtering it according to the function $\parent$, we obtain the following complexity for the function $\children$.

\begin{corollary}\label{cor:children-complexity}
	  The function $\children$ can be computed in time $O(\cP+f^{2}\cN^2)$ where $f$ is a maximum size of a solution in $\Fmax$, $\cN$ is the complexity of  function $\neig$ and $\cP$ is the time needed to compute the function $\pi$.
\end{corollary}

\begin{lemma}\label{lem:samePath}
Let $F\in \Fmax$ with $\pi(F)=f_1,...,f_\ell$, let $F_0,...,F_k$ be its canonical path and let $i \leq k$.
If $f_1,...,f_{\prox{F_i}{F}}$ is a  prefix of $\pi(F_i)$, then the canonical path of $F_i$ is $F_0,...,F_i$.
\end{lemma}

\begin{proof}

Let $T_0,...,T_k$ be the canonical path of $F_i$.
By definition we have $T_0 = F_0$.
Assume that $T_0,...,T_j=F_0,...,F_j$ for some $j < i$.
By Lemma~\ref{lem:atleastj2}, we know $\prox{F_j}{F}<\prox{F_i}{F}$, so there exists $r < \prox{F_i}{F}$ such that $\{f_1,...,f_r\} \subseteq F_j$ and $f_{r+1}\notin F_j$.

Hence, since $f_1,...,f_{r+1}$ is a prefix of $\pi(F_i)$, we have $r=\prox{F_j}{F}=\prox{F_j}{F_i}<\prox{F_i}{F}$.
Thus \begin{align*}
	F_{j+1} & =\Next(F_j,F) =\min_{<_{Lex}}\{F'\in \neig(F_j):\{j_1,...,j_{r+1}\}\subseteq F' \}\\
	& =\min_{<_{Lex}}\{F'\in \neig(T_j):\{j_1,...,j_{r+1}\}\subseteq F'\} = \Next(T_j,F_i) =T_{j+1}.
\end{align*}
This concludes the proof.
\end{proof}

The following key lemma is the generalisation of Lemma \ref{lem:prefix} in the case of chordal completions.

\begin{lemma}\label{lem:prefix3}
     Let $F_1,F_2\in \cF$ with $\pi(F_1)=f_1,...,f_\ell$ and  $\pi(F_2)=t_1,...,t_k$.
     Assume furthermore that there exists $j\leq \ell$ such that $\{f_1,...,f_j\} \subseteq F_2$, then one of the two possibilities holds:
     \begin{enumerate}
          \item $F_2 \prec_{\pi} F_1$;\label{item:prefix3-1}
          \item $f_i=t_i$ for all $i\leq j$.\label{item:prefix3-2}
      \end{enumerate}
\end{lemma}

\begin{proof}
Let $i$ be the smallest index such that $t_i\neq f_i$.
If $i>j$, then \ref{item:prefix3-2} holds and we are done.

Assume that $i\leq j$.
By minimality of $i$, we have  $F_1^{i-1}=F_2^{i-1}=:L$.
Since $f_i \in F_2$ and since $f_i\in {F_1^{i-1}}^+$, we know that $f_i \in L^+ \cap F_2$.
Therefore, as $t_i= \min\limits_{<_L}(L^{+}\cap F_2)$, it holds $t_i\leq_L f_i$, and since $t_i\neq f_i$, we have $t_i<_L f_i$.
It proves $F_2 \prec_{\pi} F_1$.
\end{proof}

The previous Lemma is the one that makes the canonical path reconstruction method possible. Independently of the prefix-closed property, if one can find a total ordering $\prec$ over $\Fmax$ which satisfies Lemma \ref{lem:prefix3}, then the canonical path reconstruction method is applicable.


	
	
	


Finally, to prove Theorem \ref{thm:principal}, it remains to show that the polynomial time computable function $\children$ defines a spanning arborescence on $\Fmax$. As already mentioned, the application of the reverse search algorithm in \cite{avis-fukuda1996}  would output all solutions  without repetition with polynomial delay and polynomial space which will conclude the proof of Theorem \ref{thm:principal}.


\begin{theorem}\label{thm:arborescence}
$\children$ defines a spanning arborescence on $\Fmax$ rooted at $F_0$.
\end{theorem}

\begin{proof}

Recall that the supergraph of solution defined by Function $\children$ is the directed graph on $\Fmax$ with arcs set $\{(F_i,F_j) \mid F_j\in \children(F_i)\}$.

 Since each $F\in \Fmax$, $F\neq F_0$ has only one in-neighbour $\parent(F)$ it is sufficient to show that for all $F\in \Fmax$, there exists a path from  $F_0$ to $F$.
 Let's denote by $\cR(F_0)\subseteq \Fmax$ the sets $F\in \Fmax$ for which there exists a path from $F_0$ to $F$.
 
 Assume for contradiction that $\cR(F_0)\neq \Fmax$ and let  $F:=\min\limits_{\prec_{\pi}}(\Fmax \setminus \cR(F_0))$.
 Let $\pi(F)=f_1,...,f_\ell$ and let $F_0,...,F_k$ be the canonical path of $F$.
 Now, consider the the smallest index $i^*\leq k$ such that $F_{i^*}\notin \cR(F_0)$.
 If $i^*=k$ then $F_{k-1}\in \cR(F_0)$ and since $F_{k-1}=\parent(F)$, $F$ would belong to $\children(F_{k-1})$ and then $F$ would belong to $\cR(F_0)$.
 So assume that $i^*<k$.
 
 Let $j:=\prox{F_{i^*}}{F}$.{}
 By minimality of $F$, we have $F\prec_{\pi}F_{i^*}$ and since $\{f_1,...,f_j\}\subseteq F_{i^*}$, by Lemma~\ref{lem:prefix3}, the $j^{\text(th)}$ prefix of $\pi(F_{i^*})$ is $f_1,...,f_j$.
 Now by Lemma~\ref{lem:samePath} the canonical path of $F_{i^*}$ is $F_0,...,F_{i^*}$, thus $F_{i^*}$ has $F_{i^*-1}$ as parent.
 By minimality of $i^*$, we know that $F_{i^*-1}\in \cR(F_0)$ and since $F_{i^*}\in \children(F_{i^*-1})$, we deduce $F_i^*\in \cR(F_0)$.

\end{proof}

\section*{Acknowledgements}
We would like to thank Aurélie Lagoutte et Lucas Pastor for helpful and fruitful discussions on the topic.


\bibliographystyle{plain}
\bibliography{biblio-triangulations}

\end{document}